\documentclass[11pt]{article}
\usepackage{amssymb}
\usepackage{amsmath}
\usepackage{amsthm}
\usepackage{breqn}
\usepackage{graphicx}
\usepackage[utf8]{inputenc}
\usepackage{amsfonts}
\usepackage{cite}
\usepackage{graphicx}
\usepackage{xcolor}
\usepackage[utf8]{inputenc}
\usepackage{float}
\usepackage{hyperref}
\def\bc{\begin{center}}
	\def\ec{\end{center}}

\def\s2c{\vskip 2cm}
\def\bt{\begin{Theorem}}
	\def\et{\end{Theorem}}
\def\bd{\begin{Definition}}
	\def\ed{\end{Definition}}
\def\bl{\begin{Lemma}}
	\def\el{\end{Lemma}}
\def\be{\begin{Example}}
	\def\ee{\end{Example}}
\def\bcor{\begin{Corollary}}
	\def\ecor{\end{Corollary}}
\def\br{\begin{Remark}}
	\def\er{\end{Remark}}

\def\mysection{\setcounter{equation}{0}\section}
\newtheorem{Lemma}{Lemma}[section]
\newtheorem{Theorem}[Lemma]{Theorem}
\newtheorem{Definition}[Lemma]{Definition}
\newtheorem{Proposition}[Lemma]{Propostion}
\newtheorem{Corollary}[Lemma]{Corollary}
\newtheorem{Remark}[Lemma]{Remark}
\date{}
\title{Age structured SIR model for the spread of infectious diseases through indirect contacts}
\author {Manoj Kumar, Syed Abbas\\
School of Basic Sciences,\\
Indian Institute of Technology Mandi,\\	Kamand (H.P.) - 175005, India
	\\Email :  sabbas.iitk@gmail.com}
	
\begin{document}
	\maketitle
	\author
	\noindent {\bf Abstract} : In this article, we discuss an age-structured SIR model in which disease not only spread through direct person to person contacts for e.g. infection due to surface contamination but it can also spread through indirect contacts. It is evident that age also plays a crucial role in SARS virus infection including COVID-19 infection. We formulate our model as an abstract semilinear Cauchy problem in an appropriate Banach space to show the existence of solution and also show the existence of steady states. It is assumed in this work that the population is in a demographic stationary state and show that there is no disease-free equilibrium point as long as there is a transmission of infection due to the indirect contacts in the environment.

\mysection{Introduction}

Infectious diseases are one of threat to humanity. Due to increase in world population and mobility, pathogen transmission is easy and it is difficult to control the spread of disease. Viral transmission depends both on the interaction with host population and with the environment.     

Mathematical models can project how infectious diseases progress. The model can suggest the possible outcome of an epidemic which will help agencies to take well though measures. In 1927, Kermack and McKendrick \cite{kermack1927contribution} introduced a model (called SIR model) by considering a given population having three compartments. The compartments are divided into individuals in susceptible $S$; infected $I$;, and removed $R$ class. It is very important to study infectious diseases and their possible nature of spread.

Most of the cases it is assumed that the spread of infectious diseases is through person to person direct contact. But some infectious diseases can also spread through indirect contacts like contact with contaminated surface having virus on it i.e. if a person touches their  eyes, mouth or nose after touching fomites or animals to human transmission. Through many studies it is observed that coronaviruses (including SARS Cov2) may persist on objects or surfaces for some hours to many days. The persistence depends on different factors (e.g. surface type,  humidity or temperature of the environment). Fomites consist of both permeable and non permeable objects or surfaces that can be contaminated with pathogenic micro-organisms and serves as a vehicle in transmission. SARS-CoV-2, the coronavirus (CoV) causing COVID-19 is creating the most severe health issues for individuals above the age of 60 — with particularly fatal results for those individuals having age above 80. In the United states, 31-59\%  of individuals ages 75 to 84 diagnosed with the virus having svere symptoms due to which hospitalization is necesaary, in comparing with 14-21\% of confirmed patients ages between 20 to 44. This data is based on US Centers for Disease Control and Prevention (CDC) report. So, it is natural to consider age structure while modeling the infectious disease transmission. The risk of transmission of infectious disease varies in different environments, for example at school, at home, at work place or in the community. \cite{RN50} studied projected age-specific contact rates for countries in different stages in development and with different demographic structures to those studied in POLYMOD (a European Commission project), which provide validated approximations to social contact patterns when directly measured data is not available. The data plotted in Fig. 1, Fig. 2, Fig. 3 and Fig. 4 show the relation between age of individual and age of contact i.e. number of contacts made by individuals at all locations, at home, school and work respectively. Yellowish color on the diagonal of Fig. 1 and Fig. 2 shows that same age individuals have more chances of direct contacts, so transmission coefficient will be large for same age individuals. \\
 \includegraphics[width=16.5cm,height=6.5cm]{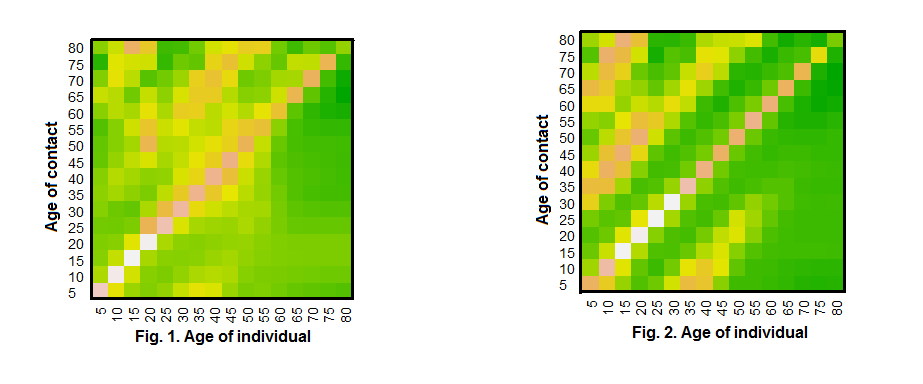} \\
 \includegraphics[width=16.5cm,height=6.5cm]{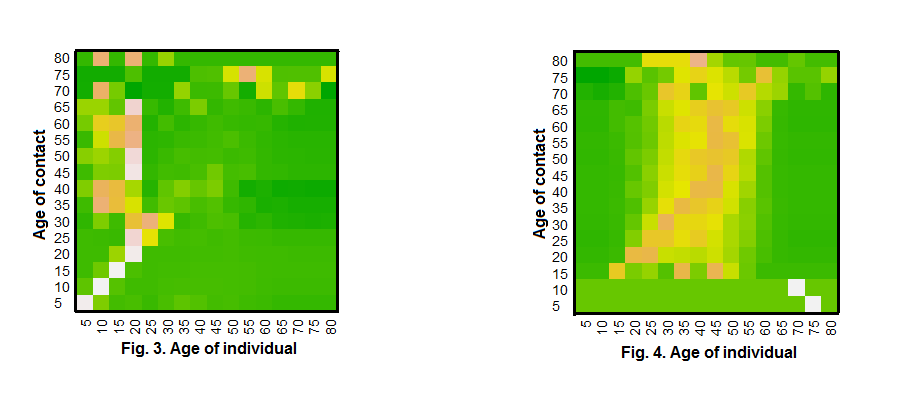} \\
 \\
 \\
 From the above heat maps, it is clear that it is natural to add age structure in ordinary differential equation(ODE) based SIR models. So, after adding age structure, the ODE based SIR models become partial differential equation (PDE) models that are more complex to analyze.
There is extensive literature available on age-structured SIR models (for more details see 
\cite{MR1057046,MR2393206,MR1846194,MR2172196, MR2515724, MR4032666, MR2813210,MR3019437, kuniya2018stability  }). In \cite{RN51} an epidemiological model which study the impact of decline in population on the dynamics of  infectious diseases especially childhood diseases is considered and also an example of measles in Italy is considered  and \cite{RN52} studied the SARS outbreak in Taiwan, using the data of daily reported cases from May 5 to June 4, 2003 to study the spread of virus.
H. Inaba \cite{MR1057046} discussed threshold and stability results for an age structured SIR model,  Andrea \cite{MR2393206} generalized the work of \cite{MR1057046} and also considered immigration of infective in all epidemiological compartments. We considered an age structured SIR model in which individuals can also get infected due to contaminated surfaces. We also assume that the net reproduction rate of the host population is unity which also makes our model different from the model considered in \cite{MR2393206}. \\
\\
Our work is divided into four sections.  In section 2, we formulate  our age structured SIR model. In section 3, we discuss the existence of solution to our model. In section 4, we discuss steady state solutions and show that there is no disease free steady state solution as long as there is transmission due to indirect contacts in the environment.

\mysection{Model Formulation}
Let $U(a,t)$ be the density of individuals of age $a$ at time $t$. $\mu(a)$ and $\beta(a)$ be age dependent mortality and fertility rates respectively. Let $a_{m}$ be the maximum age which an individual can attain i.e. the maximum life span of an individual. Then the evolution of $U(a,t)$ can be modeled by the following McKendrick-Von Foerster PDE with initial and boundary conditions:

\begin{equation} \label{2.1}
\begin{cases}
  \frac{\partial U(a,t)}{\partial t} + \frac{\partial U(a,t) }{\partial a}=  - \mu (a)U(a,t) \quad (a,t) \in (0,a_{m}) \times (0,\infty)   \\
 U(0,t)= \int_{0}^{a_{m}}  \beta(a)U(a,t)da \quad t \in (0, \infty) \\
 U(a,0)=U_{0}(a) \quad a \in (0,a_{m}),\\
\end{cases}
\end{equation}
where $U(0,t)$ denotes the number of newborns per unit time at time $t$. We suppose that the mortality rate $ \mu \in L_{loc}^{1}([0,a_{m}))$ with the condition $\int_{0}^{a_{m}} \mu(a) da= + \infty$ and the fertility rate $\beta \in L^{\infty}(0,a_{m}).$   $e^{-\int_{0}^{a} \mu(s) ds}$ indicates  the proportion of individuals who are still living at age $a$ and $\int_{0}^{a_{m}} \beta(a)e^{-\int_{0}^{a} \mu(s) ds} da$ represents the net reproduction rate. Let us assume that the net reproduction rate is $1$. So, steady state solution is given by $U(a,t)=U(a)=\beta_{0} e^{-\int_{0}^{a} \mu(\tau) d \tau}$, where $\beta_{0}$ is given by $$ \beta_{0}= \frac{\int_{0}^{a_{m}} U_{0}(a) da}{\int_{0}^{a_{m}} e^{-\int_{0}^{a}\mu(\tau) d \tau} da}. $$
Let $S(a,t), I(a,t)$ and $R(a,t)$ be the densities of susceptible, infective and recovered individuals of age $a$ at time $t$. $r(a,b)$ is the age dependent transmission coefficient which describes the contact process between susceptible and infective individuals i.e. $r(a,b)S(a,t)I(b,t) da db$ is the number of individuals who are susceptibles with age lies between $a$ and $a+da$ and contract the disease after contact with an infective individual aged between $b$ and $b+db$. We assume the form of force of infection is given in the following functional form $$ \lambda(a,t) = \int_{0}^{a_{m}} r(a,\eta) I(\eta,t) d \eta. $$ Then the disease spread according to the following system of partial differential equations
\begin{equation} \label{2.2}
\begin{cases}
  \frac{\partial S(a,t)}{\partial t} + \frac{\partial S(a,t) }{\partial a}= -\lambda(a,t)S(a,t)-c(a)S(a,t)  - \mu (a)S(a,t)    \\
  \frac{\partial I(a,t)}{\partial t} + \frac{\partial I(a,t) }{\partial a}= \lambda(a,t)S(a,t)+c(a)S(a,t)-b(a) I(a,t)  - \mu (a)I(a,t)
  \\
  \frac{\partial R(a,t)}{\partial t} + \frac{\partial R(a,t) }{\partial a}= b(a)I(a,t)  - \mu (a)R(a,t)
  \\
S(0,t)= \int_{0}^{a_{m}}  \beta(a)(S(a,t)+I(a,t)+R(a,t))da, ~~ I(0,t)=0,~~ R(0,t)=0 \\
 S(a,0)=S_{0}(a),~~I(a,0)=I_{0}(a)~\text{and}~ R(a,0)=R_{0}(a).    \\

\end{cases}
\end{equation}
 $b(a)$ is the recovery rate of individuals and $c(a)$ is the proportion of individuals which are infected due to surface contamination. This factor $c(a)$ depends on the proportion of frontline workers as they are susceptible to viral infection from indirect contacts even during lockdown situation (if lockdown is imposed). Here we are assuming that spread of disease already started and fomites are present in the environment even if transmission coefficient $r(a,b)$ is zero.  Assume that $b,c \in L^{\infty}(0,a_{m})$ and $r \in L^{\infty}((0,a_{m}) \times (0,a_{m}))$ and also assume that all are non negative. \cite{RN49} studied how respiratory and viral disease spread in the presence of fomites. It is observed that enveloped respiratory viruses remain viable for less time and the nonenveloped enteric viruses remain viable for longer time. They calculated the inactivations coefficients of various respiratory viruses. Fig. 5 shows the respiratory virus inactivation rates $(K_{j})$. In Fig. 5, we have used the short forms flu and cov for influenza and  coronavirus respectively.  \\
 \includegraphics[width=17cm,height=9.5cm]{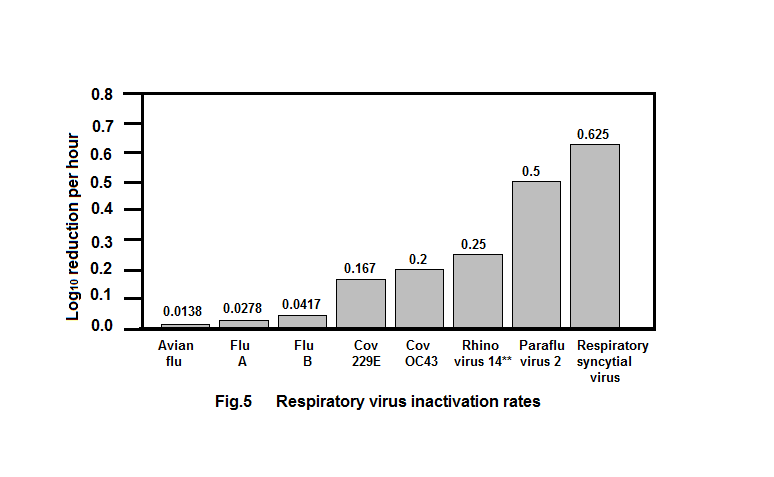}

 We impose the following conditions on our model:
 \\
 i)Although there may be incubation period for some diseases but here we are assuming that there is no incubation period and the individuals become infected instantaneously after contact with infected individuals or fomites. \\
 ii) We assume that the age zero individuals can not be infected. \\
 iii) Transmission coefficient $r(a,b)$ only summarizes the contact process between susceptible and infected individuals. \\
 iv) Population is in stationary demographic state. \\
 v) The susceptible individuals who got infected due to contact with infected individuals have not infected due to the contact with fomites and vice versa.
  \\
 Let $\overline{S}(a,t), \overline{I}(a,t)$ and $\overline{R}(a,t)$ be defined in the following way
 $$ \overline{S}(a,t)=\frac{S(a,t)}{U(a,t)}, \overline{I}(a,t)=\frac{I(a,t)}{U(a,t)}~ \text{and}~ \overline{R}(a,t)=\frac{R(a,t)}{U(a,t)} $$
 and the force of infection is given by
 $$ \lambda(a,t) = \int_{0}^{a_{m}} r(a,\eta) U(\eta) \overline{I}(\eta,t) d \eta. $$
 Then our new system becomes
 \begin{equation} \label{2.3}
\begin{cases}
  \frac{\partial \overline{S}(a,t)}{\partial t} + \frac{\partial  \overline{S}(a,t) }{\partial a}= -\lambda(a,t)\overline{S}(a,t)-c(a)\overline{S}(a,t)    \\
  \frac{\partial \overline{I}(a,t)}{\partial t} + \frac{\partial \overline{I}(a,t) }{\partial a}= \lambda(a,t)\overline{S}(a,t)+c(a)\overline{S}(a,t)-b(a) \overline{I}(a,t)
  \\
  \frac{\partial \overline{R}(a,t)}{\partial t} + \frac{\partial \overline{R}(a,t) }{\partial a}= b(a)\overline{I}(a,t)
  \\
  \overline{S}(0,t)= 1, ~~ \overline{I}(0,t)=0,~~ \overline{R}(0,t)=0 \\
 \overline{S}(a,0)=\overline{S}_{0}(a),~~\overline{I}(a,0)=\overline{I}_{0}(a)~\text{and}~ \overline{R}(a,0)=\overline{R}_{=}(a)
  \\
 \overline{S}(a,t)+\overline{I}(a,t)+\overline{R}(a,t)=1. \\
\end{cases}
\end{equation}
 So, new transformations reduced our system into a simpler form i.e. boundary conditions now become constant and there is no term involving natural mortality rate.

\mysection{Existence of solution}
If we observe system (\ref{2.3}) carefully, then it is clear that once susceptible and infected individuals are known, recovered individuals can be obtained easily so, it is enough to show the existence of solution to the below SI system instead of full SIR system
\begin{equation} \label{3.1}
\begin{cases}
  \frac{\partial \overline{S}(a,t)}{\partial t} + \frac{\partial  \overline{S}(a,t) }{\partial a}= -\lambda(a,t)\overline{S}(a,t)-c(a)\overline{S}(a,t)    \\
  \frac{\partial \overline{I}(a,t)}{\partial t} + \frac{\partial \overline{I}(a,t) }{\partial a}= \lambda(a,t)\overline{S}(a,t)+c(a)\overline{S}(a,t)-b(a) \overline{I}(a,t)
  \\
  \overline{S}(0,t)= 1, ~~ \overline{I}(0,t)=0.
  \\

\end{cases}
\end{equation}
 We will analyze the system (\ref{3.1}) only, because force of infection does not explicitly depend on recovered individuals.
Let $\tilde{S}(a,t)=\overline{S}(a,t)-1, \tilde{I}(a,t)=\overline{I}(a,t),$
then the system (\ref{3.1}) reduces to
 \begin{equation} \label{3.2}
\begin{cases}
  \frac{\partial \tilde{S}(a,t)}{\partial t} + \frac{\partial  \tilde{S}(a,t) }{\partial a}= -\lambda(a,t)(1+\tilde{S}(a,t))-c(a)(1+\tilde{S}(a,t))    \\
  \frac{\partial \tilde{I}(a,t)}{\partial t} + \frac{\partial \tilde{I}(a,t) }{\partial a}= \lambda(a,t)(1+\tilde{S}(a,t))+c(a)(1+\tilde{S}(a,t))-b(a) \tilde{I}(a,t)
  \\
  \tilde{S}(0,t)= 0, ~~ \tilde{I}(0,t)=0.
  \\
  \lambda(a,t) = \int_{0}^{a_{m}} r(a,\eta) U(\eta) \tilde{I}(\eta,t) d \eta. \\
\end{cases}
\end{equation}
Let $X=L^{1}(0,a_{m} ; \mathbb{C}^{2})$ equipped with the $L^{1}$ norm and  linear operator $\mathcal{A}$ be defined as
$$ (\mathcal{A} \xi)(a)=(-\frac{d}{d a} \xi_{1}(a)-c(a) \xi_{1}(a),-\frac{d}{da} \xi_{2}(a)-b(a) \xi_{2}(a)-c(a) \xi_{1}(a))  $$
$$ \text{where}~ \xi = (\xi_{1}(a),\xi_{2}(a)) \in D(\mathcal{A})  $$
$$ D(\mathcal{A}) = \{ \xi = (\xi_{1},\xi_{2}) \in X~|~ \xi_{1},\xi_{2} \in AC[0,a_{m}], \xi(0)=(0,0)  \} $$
$AC[0,a_{m}]$ is the set of absolutely continuous functions. \\
Suppose that $r(a,b) \in L^{\infty}((0,a_{m}) \times (0,a_{m}))$ and $$ (F \xi)(a)=(-(P \xi_{2})(a)(1+\xi_{1}(a))-c(a),(P \xi_{2})(a)(1+\xi_{1}(a))+c(a)), ~ \xi \in X,$$
where bounded linear operator $P$ is defined by
$$ (P \psi)(a) = \int_{0}^{a_{m}} r(a,\eta)U(\eta) \psi(\eta) d \eta, ~ \psi \in L^{1}(0,a_{m}). $$
Now, system (\ref{3.2}) can be written as an abstract semilinear Cauchy problem in Banach space $X$

$$ \frac{d}{dt}Z(t)= \mathcal{A}(t)Z(t)+F(Z(t)), ~ Z(0)=Z_{0} \in Z $$
$$ \text{where}~ Z(t) = (\tilde{S}(\cdot,t), \tilde{I}(\cdot,t)) \in Z, ~ Z_{0}(a)=(\tilde{S}_{0}(a),\tilde{I}_{0}(a)) $$
In the same manner as proved in \cite{MR1057046}, we can prove that $\mathcal{A}$ generates a $C_{0}$ semigroup $S(t), t \ge 0$ and $F$ is continuously Fr\'echet differentiable on $X$. \\
So, for each $Z_{0} \in X$, there exists a maximal interval of existence $[0,t_{0})$ and a unique solution $t \longrightarrow Z(t;Z_{0})$ which is continuous from $[0,t_{0})$ to $X$ such that
$$ Z(t,Z_{0})=S(t)Z_{0}+\int_{0}^{t} S(t-\sigma) F(Z(\sigma;Z_{0})) d \sigma ~~ \forall~ t \in [0,t_{0}]. $$
Moreover, if $Z_{0} \in D(\mathcal{A}),$ then $Z(t;Z_{0}) \in D(\mathcal{A})$ for $0 \le t < t_{0}$ and $t \longrightarrow Z(t;Z_{0})$ is continuously  differentiable and satisfies (\ref{3.2}) on $[0,t_{0}).$

\mysection{Steady state solutions}
\begin{equation} \label{4.1}
\begin{cases}
   \frac{d  \overline{S}(a) }{da}= -\lambda(a)\overline{S}(a)-c(a)\overline{S}(a)    \\
   \frac{d \overline{I}(a,t) }{da}= \lambda(a)\overline{S}(a)+c(a)\overline{S}(a)-b(a) \overline{I}(a)
  \\
  \overline{S}(0)= 1, ~~ \overline{I}(0)=0
  \\

\end{cases}
\end{equation}
with $ \lambda(a) = \int_{0}^{a_{m}} r(a,\eta) U(\eta)\overline{I}(\eta) d \eta $. \\
Steady state solution can be obtained as
$$\overline{S}(a) = \exp \left( -\int_{0}^{a} (\lambda(\sigma)+c(\sigma) d \sigma \right)  $$
$$ \overline{I}(a) = \int_{0}^{a} \exp \left ( -\int_{\sigma}^{a} b(\eta) d \eta \right) (\lambda(\sigma)+c(\sigma)) \exp \left( -\int_{0}^{\sigma} (\lambda(\eta)+c(\eta)) d \eta \right) d \sigma. $$ 
The force of infection depends on number of infected individuals and number of infected individuals also depend on the proportion of individuals getting infected due to indirect contacts (steady state solution shows this). So, force of infection will automatically take care of fomites present in the environment.  \\
The force of infection is given by
\begin{eqnarray}
 \label{4.2}  \lambda(a) &=& \int_{0}^{a_{m}} r(a,\zeta) U(\zeta)\overline{I}(\zeta) d \zeta
 \\  &=& \int_{0}^{a_{m}} r(a,\zeta)U(\zeta) \int_{0}^{\zeta} \exp \left( -\int_{\sigma}^{\zeta} b(\eta) d \eta \right) (\lambda(\sigma)+c(\sigma)) \exp \left( -\int_{0}^{\sigma}(\lambda(\eta)+c(\eta)) d \eta \right) d \sigma d \zeta. \nonumber \\
  \label{4.3} &=& \int_{0}^{a_{m}} \phi(a,\sigma) (\lambda(\sigma)+c(\sigma)) \exp \left( -\int_{0}^{\sigma}(\lambda(\eta)+c(\eta)) d \eta \right) d \sigma
 \\ &&
 \text{where}~ \label{4.4} \phi(a,\sigma) = \int_{\sigma}^{a_{m}} r(a,\zeta)U(\zeta) \exp \left( -\int_{\sigma}^{\zeta} b(\eta) d \eta \right) d \zeta
 \end{eqnarray}
 Using (\ref{4.2}), we can get the following estimate
 $$ | \lambda(a) | \le U \| r \|_{\infty} \| \overline{I} \|_{1} $$
 where $\| \|_{\infty}$ and $\| \|_{1}$ are the  $L^{\infty}$ and  $L^{1}$ norms respectively and $U$ is the total population.
 $$ \text{Therefore,}~ \lambda \in L^{\infty}(0,a_{m}).  $$
 It is clear that there is no disease free equilibrium as long as there is transmission due to fomites in the environment. That means if there are fomites present in the environment contaminated with pathogenic micro-organisms, disease still can spread without direct contact between susceptible and infected individuals.  \\
 On Banach space $E=L^{1}(0,a_{m})$, with positive cone $E_{+}= \{ \psi \in E ~|~ \psi \ge 0 ~ a.e. \},$ let us define
\begin{eqnarray} \label{4.5}
  \Phi(\psi)(a) = \int_{0}^{a_{m}} \phi(a,\sigma) (\psi(\sigma)+c(\sigma)) \exp \left( -\int_{0}^{\sigma}(\psi(\eta)+c(\eta)) d \eta \right) d \sigma
  \end{eqnarray}
 Suppose that we have the following assumptions
 \\
 \begin{itemize}
\item[(A1)] $r(\cdot,\cdot)$ satisfies $ \lim_{h \longrightarrow 0} \int_{0}^{a_{m}} \| r(a+h,s) - r(a,s) \| da = 0 $ uniformly for $s \in \mathbb{R}$ with $r(\cdot,\cdot)$ extended by defining $ r(a,s) = 0$ for a.e. $a,s \in (-\infty, 0) \cup (a_{m}, \infty).$ \\
 \item[(A2)] There exist $m >0, 0 < \alpha < a_{m}$ such that $r(a,b) \ge m$ for a.e. $(a,b) \in (0,a_{m}) \times (a_{m}-\alpha, a_{m})$.\\
 \item[(A3)] There exist $a_{1},a_{2}$ satisfying $0 \le a_{1} < a_{2} \le a_{m}$ such that $ c(a) >0 $ a.e. $a \in (a_{1},a_{2})$.
 \end{itemize}
 Observe that
\begin{eqnarray} \label{4.6n}
  \Phi(0)(a) = \int_{0}^{a_{m}} \phi(a,\sigma) \exp \left( - \int_{0}^{\sigma}c(\eta) d \eta \right) c(\sigma) d \sigma .
  \end{eqnarray}
 Since force of infection is non negative, we have
 $$ \lambda(a) \ge \Phi(0)(a) ~~ \text{a.e.}~ a \in (0,a_{m})                                                                                                          $$
 and because of assumption (A2), we have $ \Phi(0)(a) >0 $.
 Now, we will prove an important theorem which will help us to show the existence of fixed point to  (\ref{4.5}).
 \begin{Theorem} \label{thm 4.1}
 Let $\mathcal{D} = \{ \psi \in L_{+}^{1}(0,a_{m})~|~ \| \psi \|_{1} \le M,~\text{M is a positive constant} \} $ and suppose that the assumptions (A1)-(A3) hold, then \\
 \begin{itemize}
 \item[(a)] $\mathcal{D}$ is bounded, closed, convex and also $\Phi(\mathcal{D}) \subseteq \mathcal{D}$.
 \item[(b)] $\Phi$ is completely continuous . \\
Hence, Schauder's principle gives existence of fixed point of (\ref{4.5}).
  \end{itemize}
  \end{Theorem}
 \begin{proof}
Boundedness of set $\mathcal{D}$ is  clear and also for any $\psi_{1},\psi_{2} \in \mathcal{D}, 0 \le p \le 1$ we have $$p \psi_{1} + (1-p) \psi_{2} \in \mathcal{D}. $$ Closedness also follows from the definition of $\mathcal{D}$. 
Now, we will show that $ \Phi(\mathcal{D}) \subseteq \mathcal{D}$.
 \begin{dmath*}
 \Phi(\psi)(a) \le \| \phi \|_{\infty} \int_{0}^{a_{m}} (\psi(\sigma)+c(\sigma)) \exp \left( -\int_{0}^{\sigma}(\psi(\eta)+c(\eta)) d \eta \right) d \sigma 
 \le   M_{1} \| \phi \|_{\infty} \| c \|_{\infty} \int_{0}^{a_{m}} \exp \left( -\int_{0}^{a_{m}} \psi(\eta) d \eta \right) d \sigma +  M_{1} \| \phi \|_{\infty}\int_{0}^{a_{m}} \psi(\sigma) \exp \left( -\int_{0}^{a_{m}} \psi(\eta) d \eta \right) d \sigma 
 = M_{1} \| \phi \|_{\infty} \| c \|_{\infty} \int_{0}^{a_{m}} \exp \left( -\int_{0}^{a_{m}} \psi(\eta) d \eta \right) d \sigma+M_{1} \| \phi \|_{\infty} \left[ 1- \exp \left(- \int_{0}^{a_{m}} \psi(s) ds \right) \right]
 \end{dmath*} 
 where $M_{1}$ is an upper bound on $\exp \left( -\int_{0}^{\sigma} c(\eta) d \eta \right)$. Now , using the fact that $ | \psi |_{1} \le M$, we can easily prove that 
 $$| \Phi(\psi)(a) |_{1} \le M $$ for some generic constant $M$.
Now, 
\begin{dmath} \label{4.7}
(\Phi(\varphi_{1}))(a)-(\Phi(\varphi_{1}))(a) = \int_{0}^{a_{m}} \left[  \varphi_{1} e^{-\int_{0}^{\sigma} \varphi_{1}(\eta) d \eta} - \varphi_{2} e^{-\int_{0}^{\sigma} \varphi_{2}(\eta) d \eta} \right] \phi(a,\sigma) e^{-\int_{0}^{\sigma}c(\eta) d \eta} d \sigma + \int_{0}^{a_{m}} \left[   e^{-\int_{0}^{\sigma} \varphi_{1} (\eta) d \eta} -  e^{-\int_{0}^{\sigma} \varphi_{2}(\eta) d \eta} \right] \phi(a,\sigma) e^{-\int_{0}^{\sigma}c(\eta) d \eta} d \sigma.
\end{dmath}  
Let us firstly estimate the first integral as follows 
\begin{dmath*}
\left| \int_{0}^{a_{m}} \left[  \varphi_{1} e^{-\int_{0}^{\sigma} \varphi_{1}(\eta) d \eta} - \varphi_{2} e^{-\int_{0}^{\sigma} \varphi_{2}(\eta) d \eta} \right] \phi(a,\sigma) e^{-\int_{0}^{\sigma}c(\eta) d \eta} d \sigma \right|
\le \| \phi \|_{\infty} M_{1} \left[ e^{-\int_{0}^{a_{m}} \varphi_{2}(\eta) d \eta} -1 - e^{-\int_{0}^{a_{m}} \varphi_{1}(\eta) d \eta} +1 \right] 
 = M_{1} \| \phi \|_{\infty} \left( e^{-\| \varphi_{2} \|_{1}} - e^{- \| \varphi_{1} \|_{1}} \right) \le M_{1} \| \phi \|_{\infty} \| \varphi_{2} - \varphi_{1} \|_{1} \le M \| \varphi_{2} - \varphi_{1} \|_{1}
\end{dmath*}
where $M$ is generic constant. Similarly,
\begin{dmath*}
\left| \int_{0}^{a_{m}} \left[   e^{-\int_{0}^{\sigma} \varphi_{1} (\eta) d \eta} -  e^{-\int_{0}^{\sigma} \varphi_{2}(\eta) d \eta} \right] \phi(a,\sigma) e^{-\int_{0}^{\sigma}c(\eta) d \eta} d \sigma \right| 
\le M_{1} \|\phi \|_{\infty} \int_{0}^{a_{m}} \left[   e^{-\int_{0}^{a_{m}} \varphi_{1} (\eta) d \eta} e^{\int_{\sigma}^{a_{m}} \varphi_{1} (\eta) d \eta} - e^{-\int_{0}^{a_{m}} \varphi_{2} (\eta) d \eta} e^{\int_{\sigma}^{a_{m}} \varphi_{2} (\eta) d \eta} \right]  d \sigma
= M_{1} \| \phi \|_{\infty} \int_{0}^{a_{m}} \left[ e^{-\| \varphi_{1} \|_{1}}  e^{\int_{\sigma}^{a_{m}} \varphi_{1}(\eta) d \eta} -  e^{-\| \varphi_{2} \|_{1}}  e^{\int_{\sigma}^{a_{m}} \varphi_{2} (\eta) d \eta} \right]  d \sigma
\le M \left( e^{\| \varphi_{1} \|}-e^{\| \varphi_{2} \|} \right)
\le M \| \varphi_{2}-\varphi_{1} \|_{1}
\end{dmath*}
which proves the continuity of $\Phi$.\\
 Now we will prove that $\Phi$ is compact operator, so let us define
    $T_{1}, T_{2} : L_{+}^{1}(0,a_{m}) \longrightarrow L_{+}^{1}(0,a_{m})$ by
   \begin{eqnarray}
   & T_{1}(\psi)(a) = \int_{0}^{a_{m}} \psi(\sigma) k_{1}(a,\sigma) d \sigma \\
    & T_{2}(\psi)(a) = \int_{0}^{a_{m}} \psi(\sigma) k_{2}(a,\sigma) d \sigma
    \\
   & k_{1}(a,\sigma) = \phi(a,\sigma) \exp \left( -\int_{0}^{\sigma} c(\eta) d\eta \right),~ k_{2}(a,\sigma) = \phi(a,\sigma) c(\sigma) \exp \left( -\int_{0}^{\sigma} c(\eta) d\eta \right).
   \end{eqnarray}

   The operators $ T_{1} , T_{2} $ are linear, continuous and positive. By applying  Riesz-Fr\'echet-Kolmogorov theorem on compactness in $ L^{1} $, we can conclude that $ T_{1}, T_{2} $ are compact operators.
   Now, let us define nonlinear operators $F_{1}, F_{2} : L_{+}^{1}(0,a_{m}) \longrightarrow L_{+}^{1}(0,a_{m})$ by
   \begin{eqnarray}
   & F_{1}(\psi)(\sigma)  = \psi(\sigma) \exp \left( -\int_{0}^{\sigma} \psi(\tau) d \tau \right) \\
    & F_{2}(\psi)(\sigma)  =  \exp \left( -\int_{0}^{\sigma} \psi(\tau) d \tau \right).
   \end{eqnarray}
   Here, $F_{1}, F_{2}$ are continuous and hence $T \circ F_{1},T \circ F_{2} $ are compact operators in $L_{+}^{1}(0,a_{m})$. \\
   Therefore, $\Phi = T \circ F_{1}+T \circ F_{2} $ is compact operator.
   Hence, Schauder's principle gives existence of fixed point of (\ref{4.5}).
  \end{proof}

   Let $T = \Phi^{'}(0)$ denote the Fr\'echet derivative of $\Phi$ at 0 i.e. \\
   $$ T(\psi)(a) = \int_{0}^{a_{m}} \phi(a,\sigma) \psi(\sigma) \exp \left( -\int_{0}^{\sigma}c(\eta) d \eta \right) d \sigma  ~~ \text{for}~~ a \in (0,a_{m}), \psi \in L^{1}(0,a_{m}).$$
 Clearly,  $T$ is a positive linear, continuous and also compact operator.
   Let us define
   \begin{eqnarray} \label{4.12}
   & T_{0}(\psi)(a) = \int_{0}^{a_{m}} \phi(a,\sigma) \psi(\sigma) d \sigma  \\
  \label{4.13}   &  T_{n}(\psi)(a) = \int_{0}^{a_{m}} \phi(a,\sigma) \psi(\sigma) \exp \left( -\int_{0}^{\sigma} c_{n}(\eta) d \eta  \right) d \sigma
   \end{eqnarray}
where $c_{n}$ is the sequence of the proportion of individuals infected due to indirect contacts.   The spectral radius $(\rho(T))$ of the operator $T$ plays an important role in deciding the nature of equilibrium solutions i.e. whether disease free equilibrium solution exists or not. In our case if there is a proportion of individuals who are infected due to fomites, disease free equilibrium point will not exists.
   Our aim is to prove the following theorem: \\
   \begin{Theorem} \label{thm 4.3}
   Let $T_{0}$ be as defined in (\ref{4.12}) and $\Phi_{n}$ be analogous to $\Phi$ in which $c$ is replaced by $c_{n}$. 
   \begin{itemize}
  \item[(a)] If spectral radius $\rho(T_{0}) \le 1 $, then the sequence $\{ \psi_{n} \}$ of fixed points of $\Phi_{n}$ converges to zero.
  \item[(b)] If spectral radius of $T_{0}$ is larger than $1$,  then $\exists \gamma >0 $ such that $\| \psi_{n} \| \ge \gamma ~ \forall n \in \mathbb{N}.$
  \end{itemize}
   \end{Theorem}
Our aim is also to prove that 
$$ \lim_{n \to \infty} \rho(T_{n}) = \rho(T_{0})$$ which gives dependence of force of infection on $c_{n}$.
   Before proving the above theorem, we will prove some lemmas and also state some theorems.

   \begin{Definition}
   Let $E_{+} \subset E$ be a cone in Banach space $E$, then the cone $E_{+}$ is called total if the following set
   $$ \{ f - g : f,g \in E_{+} \} $$ is dense in  Banach space $E$.
   \end{Definition}
   \begin{Theorem} \label{thm 4.5}
(Krein-Rutman (1948))   Let $E$ be a real Banach space and $E_{+}$ be total order cone in $E$. Let $\mathbb{A} : E \longrightarrow E$ be positive linear and compact operator w.r.t. $E_{+}$ and also $\rho(\mathbb{A}) >0$. Then $\rho(\mathbb{A})$ is an eigen value of $\mathbb{A}$ and $\mathbb{A}^{*}$ with eigen vectors in $E_{+}$ and $E_{+}^{*}$ respectively.
   \end{Theorem}
   In SIR model without fomites transmission coefficient $c$, Inaba \cite{MR1057046} proved the following results:
   \begin{Theorem} \label{thm 4.6}
  (\cite{MR1057046} Proposition 4.6) Let $T$ be the Fr\'echet derivative of $\Phi$ at $0$
   \begin{itemize}
   \item[(a)] If spectral radius $\rho(T) \le 1$, then there is a disease free fixed point $\psi=0$ to the operator $\Phi$.
   \item[(b)] If spectral radius $\rho(T) >1$, then there exist atleast one non zero fixed point of $\Phi$.
   \end{itemize}
   \end{Theorem}
   \begin{Theorem} \label{thm 4.7}
  (\cite{MR0423039} Theorem V6.6) Let $ E= L^{p}(\mu) $, $p \in [1,\infty]$ and $(Z,\mathcal{S},\mu)$ be a $\sigma -$ finite measure space. Suppose
   $\mathbb{A} \in \mathcal{L}(E)$ is defined by
   $$ \mathbb{A} g(t) = \int \mathcal{K}(s,t)g(s) d \mu(s),~ g \in L^{p}(\mu),$$  non negative $\mathcal{K}$  is $\mathcal{S} \times \mathcal{S}$ measurable kernel which satisfy the following assumptions
   \begin{itemize}
   \item[(a)] Some power of $\mathbb{A}$ is compact.
   \item[(b)] $C \in \mathcal{S}$ and $\mu(C) > 0, \mu(Z \setminus C) >0 $ $$ \implies \int_{Z \setminus C} \int_{C} \mathcal{K}(s,t) d \mu(s) d \mu(t) >0.$$

\end{itemize}
Then $\rho(\mathbb{A}) > 0$ is an eigen value of $\mathbb{A}$ with a unique normalized eigen function $g$ satisfying $g(C) >0 ~ \mu$-a.e.  ; moreover if $\mathcal{K}(s,t) >0~ \mu \otimes \mu$-a.e, then every other eigen value $\lambda$ of $\mathbb{A}$ has the bound $| \lambda | < \rho(\mathbb{A}).$
   \end{Theorem}
   Let $ \{ c_{n} \} $ be a sequence in $L_{+}^{\infty}(0,a_{m})$ such that $ c_{n}(a)\longrightarrow 0$ as $ n \longrightarrow \infty$ a.e. $a \in (0,a_{m})$ i.e. proportion of individuals who are susceptible to fomite infection are becoming less and $\Phi_{0}$ is defined as in $(\ref{4.5})$ with $c=0$.
   \begin{Proposition}
   There exist a converging subsequence $\{ \psi_{n_{k}} \}$ of $\{ \psi_{n} \}$ such that if $ \psi = \lim_{k \to \infty} \psi_{n_{k}} $, then $\psi$ is the fixed point of $\Phi_{0}$.

   \end{Proposition}
 \begin{proof} Because $\Phi_{0}$ is compact and $0 \le \| \psi_{n} \| \le M, ~ \exists$ a converging subsequence
 $\{  \Phi_{0} (\psi_{n_{k}}) \}$  and let $$ \psi = \lim_{k \to \infty} \Phi_{0} (\psi_{n_{k}}). $$
 Because

$$  \Phi_{n_{k}} (\psi_{n_{k}})-\Phi_{0} (\psi_{n_{k}})= \int_{0}^{a_{m}} \phi(a,\sigma) \psi(\sigma) \exp \left( -\int_{0}^{\sigma}c_{n_{k}}(\eta) d \eta \right) d \sigma - \int_{0}^{a_{m}} \phi(a,\sigma) \psi(\sigma)  d \sigma  $$
$$ = \int_{0}^{a_{m}} \phi(a,\sigma) \psi(\sigma) \left[ \exp \left( -\int_{0}^{\sigma}c_{n_{k}}(\eta) d \eta \right) -1 \right] d \sigma $$
$$\because \lim_{n \to \infty} c_{n}(a) =0 ~ a.e. ~ a \in (0,a_{m})  $$
$$ \text{we have} ~ \lim_{k \to \infty} \left[  \Phi_{n_{k}} (\psi_{n_{k}})-\Phi_{0} (\psi_{n_{k}}) \right] =0. $$
$$ \therefore \lim_{k \to \infty} \psi_{n_{k}} = \lim_{k \to \infty} \Phi_{n_{k}}(\psi_{n_{k}}) $$
$$ = \lim_{k \to \infty} \left[\Phi_{0} (\psi_{n_{k}}) + \Phi_{n_{k}} (\psi_{n_{k}})-\Phi_{0} (\psi_{n_{k}})  \right] = \psi. $$
Because $\Phi_{0}$ is continuous, we have
$$ \lim_{k \to \infty} \Phi_{0} (\psi_{n_{k}}) = \Phi_{0}(\psi) = \psi  $$
which proves that $$  \Phi_{0}(\psi) = \psi. $$

\end{proof}
\begin{Lemma} \label{lemma 4.9}
Suppose $T_{0}$ be as defined in (\ref{4.12}), then $\rho(T_{0})$ is an eigen value of both $T_{0}$ and $T_{0}^{*}$ with unique strictly positive normalized eigen vectors $\psi$ and $f$ respectively.
\end{Lemma}

\begin{proof}  We know that
$$ T_{0}(\psi)(a) = \int_{0}^{a_{m}} \phi(a,\eta) \psi(\eta) d \eta $$
and is a compact operator by Theorem \ref{thm 4.1}.  Comparing $T_{0}$ with $\mathbb{A}$, conditions of Theorem \ref{thm 4.7} are satisfied and therefore $\rho(T_{0}) >0$ is the only eigen value of $T_{0}$ with a unique normalized eigen vector $\psi \in L_{+}^{1}(0,a_{m})$, satisfying $\psi(a) > 0 ~ a.e. $ and every other eigen value $\lambda$ of $T_{0}$ satisfy $ | \lambda| < \rho(T_{0})$. Also $T_{0}$ and  $T_{0}^{*}$ both have same non zero eigen values with same multiplicities. Since $\rho(T_{0})$ is the only eigen value of $T_{0}$ with a unique normalized eigen vector $\psi$, $\rho(T_{0})$ is also an algebraically simple eigen value of $T_{0}^{*}$ with unique normalized eigen function $f$. Now our task is to prove that eigen function $f$ is strictly positive. Suppose function $\hat{f} \in L_{+}^{\infty} \setminus \{0\}$  representing the functional $f$ be defined as
 $$ \langle f,\psi \rangle = \int_{0}^{a_{m}} \hat{f}(\eta) \psi(\eta) d \eta ~~ \forall \psi \in L^{1}(0,a_{m}).  $$
 $$ \text{Now},~ T_{0}^{*}(\varphi)(a) = \int_{0}^{a_{m}} \varphi(\eta) \phi(\eta,a) d \eta ~~ \forall  \varphi \in L^{\infty}(0,a_{m}), $$
  there exist a function  $g : [0,a_{m}] \longrightarrow \mathbb{R}$ which is continuous and $g(a)>0 ~~ \forall a \in [0,a_{m})$ and vanishes at $a_{m}$ (because of assumption (A2)) such that  $\phi(\eta,a) \ge g(a) ~~ a.e.~ \eta ,a \in (0,a_{m}). $\\
 $$ \text{Then},~ \hat{f}(a) = \frac{1}{\rho(T_{0})} T_{0}^{*}(\hat{f})(a) \ge \frac{1}{\rho(T_{0})} g(a) \int_{0}^{a_{m}} \hat{f}(\eta) d \eta >0  $$
 So, $f$ is strictly positive as $\hat{f} \in L_{+}^{\infty}(0,a_{m}) \setminus \{0\}$.
 \end{proof}
 \begin{Lemma} \label{lemma 4.10}
 Let $T_{0}$ and $T_{n}$ be as defined in  (\ref{4.12}) and (\ref{4.13}) respectively. Then $$ \lim_{n \to \infty}\rho(T_{n}) = \rho(T_{0}) ~ \text{and}~ \rho(T_{n}) \ge \rho(T_{0})~ \forall n. $$

  \end{Lemma}

\begin{proof} Clearly $T_{n} \longrightarrow T_{0} $ uniformly. Since $T_{0}$ and $T_{n}$ are compact operators and $\rho(T_{0})$ and $\rho(T_{n})$ are simple eigen values of $T_{0}$ and $T_{n}$ respectively, we have the conclusion of our lemma.
\end{proof}
\vspace{0.3cm}

Now, we are ready to prove our Theorem \ref{thm 4.3} \\
\begin{proof} We know that any converging subsequence $\{\psi_{n_{k}}\}$ of $\{ \psi_{n} \}$ converges to $\psi$, the fixed point of $\Phi_{0}$. By Theorem \ref{thm 4.6}, for $\rho(T_{0}) \le 1, \Phi_{0}$ has only one fixed point which is $0$.  So, every convergent subsequence of $\{\psi_{n}\}$ converges to zero, i.e. the sequence $\{\psi_{n} \}$ converges to zero.
Now, we will prove the part (b) of the theorem. \\
Given $\rho(T_{0}) > 1$, by lemma \ref{lemma 4.10} we have $ \rho(T_{n}) >1 ~~ \forall n .$ \\
Let $ f_{n} \in (L_{+}^{1}(0,a_{m}))^{*} \setminus \{0\} $ be the strictly positive eigen vector of $T_{n}^{*}$ with eigen value $\rho(T_{n})$. Then for all $n$, we have
$$ \langle f_{n},\psi_{n} \rangle ~=~ \langle f_{n},\Phi_{n} \psi_{n} \rangle ~=~
\langle f_{n}, \bar{\Phi}_{n} \psi_{n}+u_{n} \rangle
$$
$$\text{where}~~ \bar{\Phi}_{n} = \Phi_{n}-u_{n}~~\text{and}~ u_{n}~ \text{is defined by integral on R.H.S. of (\ref{4.6n}) with c replaced by}~ c_{n}. $$
Observe that $$ \exp(-\| \psi \|_{1}) T_{0} \psi \le \bar{\Phi} \psi \le T_{0} \psi~~~ \forall \psi \in L_{+}^{1}(0,a_{m})  $$
Therefore,
$$  \langle f_{n},\psi_{n} \rangle \ge \langle f_{n},\exp(-\| \psi_{n} \|_{1}) T_{n} \psi_{n}+u_{n} \rangle >  \langle f_{n},\exp(-\| \psi_{n} \|_{1}) T_{n} \psi_{n} \rangle $$
$$  =\exp(-\| \psi_{n} \|_{1}) \langle T_{n}^{*}f_{n},\psi_{n} \rangle = \exp(-\| \psi_{n} \|_{1}) \rho(T_{n}) \langle f_{n},\psi_{n} \rangle   $$
Therefore, $$ \exp(-\| \psi_{n} \|_{1}) \rho(T_{n}) <1~~ \forall n $$
$$ i.e. ~~ \| \psi_{n} \|_{1} > \log(\rho(T_{n})) \ge \log(\rho(T_{0})) $$
$$ \text{choose}~~ \gamma = \log(\rho(T_{0})) >0 $$
then the conclusion of our theorem holds. \end{proof}

\mysection{Discussion}	
The figures on data related to interaction show that the age plays a crucial role in SARS diseases and especially in COVID-19 infection as well as in recovery. So, we have studied an age structured SIR model in which susceptible individuals not only get infected due to direct contact with infected person, but can also get infected due to contact with contaminated surfaces.  We proved that there is no disease free equilibrium as long as there is transmission due to indirect contacts in the environment. That means for instance if there are fomites present in the environment contaminated with pathogenic micro-organisms, disease still can spread without direct contact between susceptible and infected individuals. So, removing fomites present on the surfaces is one of the effective measure to slow the infection. Hence sanitization of surfaces and proper care to frontline workers will help to fight with such diseases.


\bibliographystyle{amsplain}

\begin{thebibliography}{10}

\bibitem{MR1057046}
H.~Inaba.
\newblock \href{https://doi.org/10.1007/BF00178326}{Threshold and stability
  results for an age-structured epidemic model}.
\newblock {\em J. Math. Biol.}, 28(4):411--434, 1990.

\bibitem{MR2393206}
Andrea Franceschetti and Andrea Pugliese.
\newblock Threshold behaviour of a {SIR} epidemic model with age structure and
  immigration.
\newblock {\em J. Math. Biol.}, 57(1):1--27, 2008.

\bibitem{kermack1927contribution}
William~Ogilvy Kermack and A.G McKendrick.
\newblock A contribution to the mathematical theory of epidemics.
\newblock {\em Proceedings of the Royal Society of London. Series A:
  Mathematical, Physical and Engineering Sciences}, 115(772):700--721, 1927.

\bibitem{RN50}
K.~Prem, A.~R. Cook, and M.~Jit.
\newblock Projecting social contact matrices in 152 countries using contact
  surveys and demographic data.
\newblock {\em Plos Computational Biology}, 13(9), 2017.


\bibitem{MR1846194}
Xue-Zhi Li, Geni Gupur, and Guang-Tian Zhu.
\newblock Threshold and stability results for an age-structured {SEIR} epidemic
  model.
\newblock {\em Comput. Math. Appl.}, 42(6-7):883--907, 2001.

\bibitem{MR2172196}
Hisashi Inaba.
\newblock Mathematical analysis of an age-structured {SIR} epidemic model with
  vertical transmission.
\newblock {\em Discrete Contin. Dyn. Syst. Ser. B}, 6(1):69--96, 2006.

\bibitem{MR2515724}
Xue-Zhi Li and Bin Fang.
\newblock Stability of an age-structured {SEIR} epidemic model with infectivity
  in latent period.
\newblock {\em Appl. Appl. Math.}, 4(1):218--236, 2009.

\bibitem{MR4032666}
Kento Okuwa, Hisashi Inaba, and Toshikazu Kuniya.
\newblock Mathematical analysis for an age-structured {SIRS} epidemic model.
\newblock {\em Math. Biosci. Eng.}, 16(5):6071--6102, 2019.

\bibitem{MR2813210}
Toshikazu Kuniya.
\newblock Global stability analysis with a discretization approach for an
  age-structured multigroup {SIR} epidemic model.
\newblock {\em Nonlinear Anal. Real World Appl.}, 12(5):2640--2655, 2011.

\bibitem{MR3019437}
Andrey~V. Melnik and Andrei Korobeinikov.
\newblock Lyapunov functions and global stability for {SIR} and {SEIR} models
  with age-dependent susceptibility.
\newblock {\em Math. Biosci. Eng.}, 10(2):369--378, 2013.

\bibitem{kuniya2018stability}
Toshikazu Kuniya.
\newblock Stability analysis of an age-structured sir epidemic model with a
  reduction method to odes.
\newblock {\em Mathematics}, 6(9):147, 2018.

\bibitem{RN51}
P.~Manfredi and J.~R. Williams.
\newblock Realistic population dynamics in epidemiological models: the impact
  of population decline on the dynamics of childhood infectious diseases -
  measles in italy as an example.
\newblock {\em Mathematical Biosciences}, 192(2):153--175, 2004.

\bibitem{RN52}
Y.~H. Hsieh, C.~W.~S. Chen, and S.~B. Hsu.
\newblock {SARS} outbreak, {Taiwan}, 2003.
\newblock {\em Emerging Infectious Diseases}, 10(2):201--206, 2004.

\bibitem{RN49}
S.~A. Boone and C.~P. Gerba.
\newblock Significance of fomites in the spread of respiratory and enteric
  viral disease.
\newblock {\em Applied and Environmental Microbiology}, 73(6):1687--1696, 2007.

\bibitem{MR0423039}
Helmut~H. Schaefer.
\newblock {\em Banach lattices and positive operators}.
\newblock Springer-Verlag, New York-Heidelberg, 1974.

\end{thebibliography}

\end{document}